\documentclass[a4paper,UKenglish,cleveref, autoref]{lipics-v2019} 

\allowdisplaybreaks[1]

\newtheorem{fact}[theorem]{Fact}
\newtheorem{conjecture}[theorem]{Conjecture}

\usepackage{amssymb}
\usepackage{bm}
\usepackage{mathtools}
\usepackage{stmaryrd}

\newcommand{\N}{\ensuremath{\mathbb{N}}}
\newcommand{\Z}{\ensuremath{\mathbb{Z}}}

\newcommand{\NP}{\ensuremath{\mathsf{NP}}}
\newcommand{\W}[1]{\ensuremath{\mathsf{W[#1]}}}

\newcommand{\timesc}{\ensuremath{\times_H}}

\newcommand{\homs}[2]{\mbox{\ensuremath{\mathsf{Hom}(#1 \to #2)}}}

\newcommand{\embs}[2]{\mbox{\ensuremath{\mathsf{Emb}(#1 \to #2)}}}

\newcommand{\subs}[2]{\mbox{\ensuremath{\mathsf{Sub}(#1 \to #2)}}}

\newcommand{\strembs}[2]{\mbox{\ensuremath{\mathsf{StrEmb}(#1 \to #2)}}}

\newcommand{\indsubs}[2]{\mbox{\ensuremath{\mathsf{IndSub}(#1 \to #2)}}}
\newcommand{\auts}[1]{\ensuremath{\mathsf{Aut}(#1)}}
\newcommand{\cphom}{\ensuremath{\mathsf{cp}\text{-}\mathsf{Hom}}}
\newcommand{\cphoms}[2]{\ensuremath{\mathsf{cp}\text{-}\mathsf{Hom}}(#1 \to #2)}

\newcommand{\cpsubs}[2]{\ensuremath{\mathsf{cp}\text{-}\mathsf{Sub}}(#1 \to #2)}
\newcommand{\cpemb}{\ensuremath{\mathsf{cp}\text{-}\mathsf{Emb}}}

\newcommand{\cpstrembs}[2]{\ensuremath{\mathsf{cp}\text{-}\mathsf{StrEmb}}(#1 \to #2)}

\newcommand{\cpindsubs}[2]{\ensuremath{\mathsf{cp}\text{-}\mathsf{IndSub}}(#1 \to #2)}

\newcommand{\clique}{\ensuremath{\textsc{Clique}}}

\newcommand{\homsprob}{\ensuremath{\textsc{Hom}}}
\newcommand{\cphomsprob}{\ensuremath{\textsc{cp-Hom}}}
\newcommand{\indsubsprob}{\ensuremath{\textsc{IndSub}}}
\newcommand{\cpindsubsprob}{\ensuremath{\textsc{cp-IndSub}}}

\newcommand{\fptred}{\ensuremath{\leq^{\mathrm{fpt}}_{\mathrm{T}}}}

\title{Counting Induced Subgraphs:\newline An Algebraic Approach to \#W[1]-hardness}
\titlerunning{Counting Induced Subgraphs: An Algebraic Approach to \#W[1]-hardness\footnote{All authors of this paper are students.}}

\author{Julian Dörfler}{Saarbrücken Graduate School of Computer Science, Saarland Informatics Campus (SIC), Germany}{s8judoer@stud.uni-saarland.de}{https://orcid.org/0000-0002-0943-8282}{}

\author{Marc Roth}{Cluster of Excellence (MMCI), Saarland Informatics Campus (SIC), Saarbrücken, Germany}{mroth@mmci.uni-saarland.de}{https://orcid.org/0000-0003-3159-9418}{}

\author{Johannes Schmitt}{ETH Z\"urich, Switzerland \and \url{https://people.math.ethz.ch/~schmittj/} }{johannes.schmitt@math.ethz.ch}{https://orcid.org/0000-0001-5774-3508}{The third author has received funding from the European Research Council (ERC) under the EU’s
Horizon 2020 research and innovation programme (No~786580).}

\author{Philip Wellnitz}{Max Planck Institute for Informatics, Saarland Informatics Campus (SIC), Saarbrücken, Germany \and \url{https://people.mpi-inf.mpg.de/~wellnitz/}}{wellnitz@mpi-inf.mpg.de}{https://orcid.org/0000-0002-6482-8478}{}

\authorrunning{J. Dörfler and M. Roth and J. Schmitt and P. Wellnitz}

\Copyright{Julian Dörfler and Marc Roth and Johannes Schmitt and Philip Wellnitz}



\acknowledgements{We are very grateful to Radu Curticapean and Holger Dell for fruitful discussions and valuable feedback on early drafts of this work.}

\nolinenumbers 

\hideLIPIcs  

\EventEditors{John Q. Open and Joan R. Access}
\EventNoEds{2}
\EventLongTitle{42nd Conference on Very Important Topics (CVIT 2016)}
\EventShortTitle{CVIT 2016}
\EventAcronym{CVIT}
\EventYear{2016}
\EventDate{December 24--27, 2016}
\EventLocation{Little Whinging, United Kingdom}
\EventLogo{}
\SeriesVolume{42}
\ArticleNo{23}

\begin{document}

\maketitle

\begin{abstract}
    We study the problem $\#\indsubsprob(\Phi)$ of counting all induced subgraphs of size $k$ in a graph $G$ that satisfy the property $\Phi$. This problem was introduced by Jerrum and Meeks and shown to be $\#\W{1}$-hard when parameterized by $k$ for some families of properties $\Phi$ including, among others, connectivity~[JCSS~15] and even- or oddness of the number of edges [Combinatorica 17]. Very recently~[IPEC~18], two of the authors introduced a novel technique for the complexity analysis
    of $\#\indsubsprob(\Phi)$, inspired by the ``topological approach to evasiveness'' of Kahn, Saks and Sturtevant~[FOCS~83] and the framework of graph motif parameters due to Curticapean, Dell and Marx~[STOC~17], allowing them to prove hardness of a wide range of properties $\Phi$. In this work, we refine this technique for graph properties that are non-trivial on edge-transitive graphs with a prime power number of edges. In particular, we fully classify the case of monotone bipartite graph
    properties:
    It is shown that, given \emph{any} graph property $\Phi$ that is closed under the removal of vertices and edges, and that is non-trivial for bipartite graphs, the problem $\#\indsubsprob(\Phi)$ is $\#\W{1}$-hard and cannot be solved in time $f(k)\cdot n^{o(k)}$ for any computable function $f$, unless the Exponential Time Hypothesis fails. This holds true even if the input graph is restricted to be bipartite and counting is done modulo a fixed prime. A similar result is shown for properties that are closed under the removal of edges only.
\end{abstract}

\section{Introduction}
The study of the computational complexity of counting problems was initiated by Valiant's seminal work about the complexity of computing the permanent~\cite{Valiant79}. In contrast to a decision problem which requires to \emph{verify} the existence of a solution, a counting problem asks to compute the \emph{number} of solutions.
Counting complexity theory is particularly interesting for problems whose decision versions are solvable efficiently but whose counting versions are intractable. One such example is the problem of finding/counting perfect matchings, whose decision version is solvable in polynomial time~\cite{Edmonds65} and whose counting version is as least as hard as every problem in the Polynomial Hierarchy $\mathsf{PH}$ with respect to polynomial-time Turing reductions~\cite{Valiant79,Toda91}.
In this work, we consider the following problem which was first introduced by Jerrum and Meeks~\cite{JerrumM15}:
Fix a graph property $\Phi$, given a graph $G$ and a positive integer $k$, compute the number of all induced subgraphs of $G$ with $k$ vertices that satisfy~$\Phi$.
We denote this problem by $\#\indsubsprob(\Phi)$ and remark that, strictly speaking, $\#\indsubsprob(\Phi)$ is the \emph{unlabeled} version of $p\text{-}\#\textsc{InducedSubgraphWithProperty}(\Phi)$ as defined in~\cite[Section~1.3.1]{JerrumM17}. In particular, our properties only depend on the isomorphism type of a graph and not on any labeling of the vertices.

We study the \emph{parameterized complexity} of $\#\indsubsprob(\Phi)$ depending on the property~$\Phi$. The underlying framework, known as \emph{parameterized counting complexity theory}, was introduced independently by Flum and Grohe~\cite{FlumG04} and McCartin~\cite{McCartin06}, and constitutes a hybrid of (classical) computational counting and parameterized complexity theory.
Here, the method of parameterization allows us to perform a multivariate analysis of the complexity of $\#\indsubsprob(\Phi)$: Instead of the distinction between polynomial-time solvable and $\NP$-hard cases, we search for properties~$\Phi$ for which the problem is solvable in time~$f(k)\cdot n^{O(1)}$, where $n$ is the number of vertices of the graph and~$f$ can be any computable function.
If this is the case, the problem is called \emph{fixed-parameter tractable}. Unfortunately, the only known cases of $\Phi$ for which $\#\indsubsprob(\Phi)$ is fixed-parameter tractable are trivial in the sense that there are only finitely many $k$ such that $\Phi$ is neither true nor false on the set of all graphs with $k$ vertices.
On the contrary, it is easy to see that $\#\indsubsprob(\Phi)$ is most likely not fixed-parameter tractable if~$\Phi$ encodes a problem whose decision version is already known to be hard. An example of the latter is the property of being a complete graph. In this case, the problem $\#\indsubsprob(\Phi)$ is identical to the problem of counting cliques of size $k$, for which even the decision version, that is, \emph{finding} a clique of size $k$ in a graph with $n$ vertices, cannot be done in time $f(k)\cdot n^{o(k)}$, unless the Exponential Time Hypothesis fails~\cite{Chenetal05,Chenetal06}.

The first non-trivial hardness result of $\#\indsubsprob(\Phi)$ was given by Jerrum and Meeks for~$\Phi$ the property of being connected~\cite{JerrumM15}. Note that, in this case, the decision version of the problem can be solved efficiently as, on input $G$ and $k$, one only has to decide whether there exists a connected component of $G$ of size at least $k$. This result initiated a line of research in which Jerrum and Meeks proved fixed-parameter tractability of $\#\indsubsprob(\Phi)$ to be unlikely for the property of having an even (or odd) number of edges~\cite{JerrumM17}, for properties that induce low edge densities~\cite{JerrumM15b} and for properties that are closed under the addition of edges and whose (edge-)minimal elements have large treewidth~\cite{Meeks16}.
More precisely, all of those results established hardness for the parameterized complexity class $\#\W{1}$, which can be seen as the parameterized counting equivalent of $\mathsf{NP}$.
In a recent breakthrough result~\cite{CurticapeanDM17}, Curticapean, Dell and Marx have shown, that for every graph property $\Phi$, the problem $\#\indsubsprob(\Phi)$ is either fixed-parameter tractable or hard for $\#\W{1}$, that is, there are no cases of intermediate difficulty.
On the downside, they did not provide an explicit criterion for $\#\W{1}$-hardness that allows to pin down the complexity of $\#\indsubsprob(\Phi)$, given a concrete property~$\Phi$.  \linebreak
\noindent However,~combining the framework of~\cite{CurticapeanDM17} with tools from the ``topological approach to evasiveness'' by Kahn, Saks and Sturtevant~\cite{KahnSS84}, two of the authors of the current paper established $\#\W{1}$-hardness for a wide range of properties, including, for example, all non-trivial properties that are closed under the removal of edges and false on odd cycles~\cite{RothS18}.
Taken together, the above results suggest the following conjecture.
\begin{conjecture}\label{con:main_conjecture}
Let $\Phi$ be a computable graph property satisfying that there are infinitely many positive integers $k$ such that $\Phi$ is neither true nor false on all graphs with $k$ vertices. Then $\#\indsubsprob(\Phi)$ is $\#\W{1}$-hard.
\end{conjecture}
Unfortunately, a proof of this conjecture seems to be a long way off. In this work however, building up on~\cite{CurticapeanDM17,RothS18}, we introduce an algebraic approach that allows us to resolve the above conjecture in case of \emph{all} non-trivial monotone properties on bipartite graphs. In particular, we obtain a matching lower bound under the Exponential Time Hypothesis.

\paragraph*{Results and techniques}
We call a graph property \emph{monotone} if it closed under the removal of vertices and edges and \emph{edge-monotone} if it is closed under the removal of edges only. Furthermore, we write $\mathsf{IS}_k$ for the graph consisting of $k$ isolated vertices and $K_{t,t}$ for the complete bipartite graph with $t$ vertices on each side. Our main theorems read as follows.
\begin{theorem}\label{thm:intro_new}
	Let $\Phi$ be a computable graph property and let $\mathcal{K}$ be the set of all prime powers $t$ such that $\Phi(\mathsf{IS}_{2t}) \neq \Phi(K_{t,t})$. If $\mathcal{K}$ is infinite then $\#\indsubsprob(\Phi)$ is $\#\W{1}$ hard. If additionally $\mathcal{K}$ is dense then it cannot be solved in time $f(k)\cdot n^{o(k)}$ for any computable function $f$ unless ETH fails. This holds true even if the input graphs to $\#\indsubsprob(\Phi)$ are restricted to be bipartite.
\end{theorem}
In the previous theorem, a set $\mathcal{K}$ is \emph{dense} if there exists a constant $c$ such that for every $m\in \mathbb{N}$, there exists a $k \in \mathcal{K}$ such that $m\leq k \leq cm$. While the hypotheses of Theorem~\ref{thm:intro_new} sound technical, the theorem applies in many situations. In particular, it is applicable to properties that are neither (edge-) monotone nor the complement thereof: Let $\Phi$ be the property of being Eulerian. The graph $K_{t,t}$ contains an Eulerian cycle if $t=2^s$ for $s\geq 1$. Hence we can apply Theorem~\ref{thm:intro_new} with $\mathcal{K}=\{2^s~|~s \geq 1\} $, which is infinite and dense. 
\begin{corollary}
	Let $\Phi$ be the property of being Eulerian. Then $\#\indsubsprob(\Phi)$ is $\#\W{1}$-hard and cannot be solved in time $f(k)\cdot n^{o(k)}$ for any computable function $f$ unless the ETH fails. This holds true even if the input graphs to $\#\indsubsprob(\Phi)$ are restricted to be bipartite.
\end{corollary}
In case $\Phi$ is edge-monotone, the condition $\Phi(\mathsf{IS}_{2t}) \neq \Phi(K_{t,t})$ is equivalent to non-triviality and if $\Phi$ is monotone, we obtain the following, more concise statement of the hardness result.
\begin{theorem}\label{thm:intro_main_bip}
Let $\Phi$ be a computable monotone graph property such that $\Phi$ and $\neg \Phi$ hold on infinitely many bipartite graphs. Then $\#\indsubsprob(\Phi)$ is $\#\W{1}$-hard and cannot be solved in time $f(k)\cdot n^{o(k)}$ for any computable function $f$ unless the Exponential Time Hypothesis fails. This holds true even if the input graphs to $\#\indsubsprob(\Phi)$ are restricted to be bipartite.
\end{theorem}

Let us illustrate further consequences of the previous theorems with respect to (edge-) monotone properties. First of all, most of the prior hardness results (\cite{JerrumM15,JerrumM15b,Meeks16,JerrumM17,RothS18}) are shown to hold in the restricted case of bipartite graphs. We provide three examples:
\begin{corollary}
	The problem $\#\indsubsprob(\Phi)$, restricted to \emph{bipartite input graphs}, is $\#\W{1}$-hard and cannot be solved in time $f(k)\cdot |V(G)|^{o(k)}$ for any computable function $f$ unless ETH fails, if $\Phi$ is one of the properties of being disconnected, planar or non-hamiltonian.
\end{corollary}
One example of a monotone property $\Phi$ for which the complexity of $\#\indsubsprob(\Phi)$ was unknown, even for general graphs, is given by the following corollary of Theorem~\ref{thm:intro_main_bip}.
\begin{corollary}
	Let $F$ be a fixed bipartite graph with at least one edge and define $\Phi(G)=1$ if $G$ does not contain a subgraph isomorphic to $F$. Then $\#\indsubsprob(\Phi)$ is $\#\W{1}$-hard and cannot be solved in time $f(k)\cdot |V(G)|^{o(k)}$ for any computable function $f$ unless ETH fails. This holds true even if the input graphs of $\#\indsubsprob(\Phi)$ are restricted to be bipartite.
\end{corollary}
As the number of induced subgraphs of size $k$ that satisfy $\Phi$ equals $\binom{|V(G)|}{k}$ minus the number of induced subgraphs of size $k$ that satisfy $\neg \Phi$, \emph{all of the previous result remain true for the complementary properties} $\neg \Phi$.

In proving the previous theorems we build up on the approach in~\cite{CurticapeanDM17,RothS18}, where it was shown that, given a graph property~$\Phi$ and a positive integer~$k$, the number of induced subgraphs of size~$k$ in a graph~$G$ that satisfy~$\Phi$ can equivalently be expressed as the following sum over all (isomorphism types of) graphs~$H$:
\begin{equation}\label{eq:intro_homs}
\sum_H a_{\Phi}(H) \cdot \#\homs{H}{G} \,,
\end{equation}
where $a_\Phi$ is a function from graphs to integers with finite support and $\#\homs{H}{G}$ is the number of graph homomorphisms from~$H$ to~$G$.
It is known that computing a linear combination of homomorphism numbers, as in the above expression,
is \emph{precisely as hard as} computing its hardest term with a non-zero coefficient (\cite{CurticapeanDM17}, also implicitly proved in~\cite{ChenM16}). We refer to this property as \emph{complexity monotonicity}.
In~\cite{RothS18} two of the authors of the current paper used a topological approach to analyze the coefficient $a_{\Phi}(K_k)$ of the complete graph on $k$ vertices.
If this coefficient is non-zero then complexity monotonicity implies that computing the number of induced subgraphs of size $k$ in a graph $G$ that satisfy $\Phi$ is at least as hard as computing the number $\#\homs{K_k}{G}$.
This, in turn, is equivalent to computing the number of cliques of size $k$ in $G$, the canonical $\#\W{1}$-complete problem~\cite{FlumG04}.
While this approach led to hardness proofs for a wide range of properties $\Phi$, it seems that resolving Conjecture~\ref{con:main_conjecture}, even restricted to monotone properties, requires a significant amount of new ideas.
Without going too much into the details\footnote{Readers familiar with~\cite{RothS18} might recall that fixed points of group actions have been used to derive a simpler formula to compute the number $a_{\Phi}(K_t)$ modulo a prime $p$ for positive powers $t$ of $p$. This formula would simplify greatly if the group had a $p$-power number of elements \emph{and} acted transitively on the edges of $K_t$.
Unfortunately, this can never happen for $t\geq 4$, since the number of edges of $K_t$ is not itself a $p$-power.} of~\cite{RothS18},
our analysis of $a_{\Phi}(K_k)$ is complicated by the fact that the number of edges of the complete graph on $k\geq 4$ vertices is not a prime power.
In this work, we hence focus on the coefficient of $a_{\Phi}(H)$ for graphs $H$ that have a prime power number of edges and for which computing $\#\homs{H}{G}$ is hard.
One example of such graphs is the biclique $K_{t,t}$ for some prime power $t$. Here a biclique $K_{t,t}$, also called a complete bipartite graph, has $t$ vertices on each side and contains every edge from a vertex on the left side to a vertex to the right side. Hence the number of edges is $t^2$ which is a prime power if $t$ is.

In analyzing the coefficient $a_{\Phi}(K_{t,t})$ of the complete bipartite graph, we invoke the results of Rivest and Vuillemin~\cite{rivestvuillemin} who considered transitive boolean functions over a domain of prime power cardinality to resolve the asymptotic version of what is known as \emph{Karp's evasiveness conjecture} (we recommend Miller's survey~\cite{Miller13} for an excellent overview).\pagebreak

\noindent Given a property~$\Phi$ and a graph~$H$, the \emph{alternating enumerator} of $\Phi$ and $H$ is defined to be
\begin{equation*}
\hat{\chi}(\Phi,H) := \sum_{S \subseteq E(H)} \Phi(H[S])\cdot (-1)^{\#S}\,,
\end{equation*}
where $H[S]$ is the graph with vertices $V(H)$ and edges $S$.
Roughly speaking, it will turn out that the value of $a_{\Phi}(H)$ is closely related to $\hat{\chi}(\Phi,H)$.
We furthermore point out that, in case $\Phi$ is closed under the removal of edges, the alternating enumerator
$\hat{\chi}(\Phi,H)$ equals what is called the reduced Euler characteristic of the simplicial complex on $E(H)$ associated to~$\Phi$~\cite{Miller13,RothS18}.
In Section~\ref{sec:alt_enum} we study the alternating enumerator in case of edge-transitive graphs, that is, graphs whose automorphism groups act transitively on the set of edges.
We give a self-contained proof of the following fact, which implicitly follows from~\cite{rivestvuillemin}.
\begin{lemma}\label{lem:alt_enum_intro}
    Let $\Phi$ be a graph property and let $H$ be an edge-transitive graph with $p^k$ edges such that $p$ is a prime and $\Phi(H[\emptyset]) \neq \Phi(H)$. Then it holds that $\hat{\chi}(\Phi,H) = (\pm 1) \mod p \,.$
\end{lemma}
Now, intuitively, Lemma~\ref{lem:alt_enum_intro} induces a strategy towards proving hardness of $\#\indsubsprob(\Phi)$:
Assume a family of edge-transitive graphs $\mathcal{H}$ can be found such that $\#E(H)$ is a prime power and $\Phi(H[\emptyset]) \neq \Phi(H)$ for every $H\in \mathcal{H}$.
Then $\#\indsubsprob(\Phi)$ is at least as hard as counting homomorphisms from graphs in $\mathcal{H}$, the latter of which is fully understood~\cite{DalmauJ04}.
This observation gives a strong motivation for the study of edge-transitive graphs with a prime power number of edges. In the second part of Section~\ref{sec:alt_enum}, we fully classify those graphs as subgraphs of bipartite graphs or vertex-transitive subgraphs of wreath graphs; consult Section~\ref{sec:alt_enum} for the formal definitions.
The proof of the following theorem, which might be of independent interest, relies on a non-trivial application of Sylow's theorems.

\begin{theorem} \label{thm:pedgetransgraphs_intro}
    Let $G$ be a connected edge-transitive graph with $p^t$ edges for some prime $p$ and positive integer $t$. Then either $G$ is bipartite or $G$ is vertex-transitive and can be obtained from the wreath graph $W_{p^k}$ for $k \geq 1$ by removing edges (or both).
\end{theorem}

With the analysis of $\hat{\chi}$ and edge-transitive graphs completed, we turn to the reduction from counting homomorphisms in Section~\ref{sec:main_reductions}.
More precisely, given a class $\mathcal{H}$ of edge-transitive graphs with a prime power number of edges and a graph property $\Phi$ such that for every $H\in \mathcal{H}$ we have that $\Phi(H[\emptyset]) \neq \Phi(H)$, we construct a parameterized Turing reduction from $\#\homsprob(\mathcal{H})$ to $\#\indsubsprob(\Phi)$.
Here, the problem $\#\homsprob(\mathcal{H})$ is defined as follows: Given as input a graph
$H\in \mathcal{H}$ and a graph $G$, compute the number of homomorphisms from $H$ to $G$.
For technical reasons, we cannot immediately transform the number of induced subgraphs that satisfy~$\Phi$ to a linear combination of homomorphism numbers as in Equation~(\ref{eq:intro_homs}).
We~solve this technical issue by introducing color-prescribed variants of those problems in an intermediate step. In this context we consider $H$-colored graphs.
Recall that a graph $G$ is $H$-colored if it comes with a homomorphism $c$ from $G$ to $H$.
A homomorphism from $H$ to $G$ is then called color-prescribed if it maps every vertex $v$ of $H$ to a vertex $u$ of $G$ satisfying that $c(u)=v$.
We demonstrate that, given an $H$-colored graph $G$ and oracle access to $\#\indsubsprob(\Phi)$, the following linear combination can be computed in time $f(|V(H)|)\cdot |V(G)|^{O(1)}$.
\begin{equation}\label{eq:intro_two}
\sum_{S \subseteq E(H)} \hat{a}_\Phi(S) \cdot \#\cphoms{H[S]}{G}.
\end{equation}
Here $\cphoms{H[S]}{G}$ denotes the set of color-prescribed homomorphisms from $H[S]$ to~$G$ and $\hat{a}_\Phi$ is a function of finite support only depending in $\Phi$.
In particular, $\hat{a}_\Phi(E(H))$ and $\hat{\chi}(\Phi,H)$ are proved to agree up to a factor of $-1$.
Finally, we establish complexity monotonicity for linear combinations of color-prescribed homomorphisms as in Equation~(\ref{eq:intro_two}), which in combination with Lemma~\ref{lem:alt_enum_intro} yields the desired reduction.\pagebreak

\noindent Combining the previous results, we invoke the reduction on graph properties that are non-trivial on bipartite graphs and prove Theorem~\ref{thm:intro_new} and Theorem~\ref{thm:intro_main_bip}, in Section~\ref{sec:apply_bipartite}.
Furthermore, we illustrate in the appendix that our algebraic approach readily extends to modular counting by proving that both, Theorem~\ref{thm:intro_new} and Theorem~\ref{thm:intro_main_bip} remain true in case counting is done modulo a fixed prime. 

\section{Preliminaries}
Given a positive integer $k$, we write $[k]$ for the set $\{1,\dots,k\}$ and given a set $A$ we write $\binom{A}{k}$ for the set of all subsets of size $k$ of $A$. Furthermore, assuming that $A$ is finite, we write $\#A$ or $|A|$ for its cardinality.
Given a function $g: A\times B \rightarrow C$ and an element $a\in A$, we write $g(a,\star)$ for the function which maps $b \in B$ to~$g(a,b)$.

\subsection{Graph theory}
Graphs in this work are considered simple, undirected and without self-loops. More precisely, a graph $G$ is a pair of a finite set $V(G)$ of vertices and a symmetric and irreflexive relation $E(G)\subseteq V(G)^2$. If a graph $H$ is obtained from $G$ by deleting a set of edges and a set of vertices of $G$, including incident edges, then $H$ is called a \emph{subgraph} of~$G$. Given a subset $\hat{V}$ of $V(G)$ we write $G[\hat{V}]$ for the graph with vertices~$\hat{V}$ and edges $E \cap \hat{V}^2$. The resulting graph is called an \emph{induced subgraph} of~$G$. An \emph{edge-subgraph} of a graph $H$ is a graph obtained from~$H$ by deleting edges. Given a set $S\subseteq E(H)$ we write $H[S]$ for the edge-subgraph $(V(H),S)$ of~$H$.

\paragraph*{Homomorphisms and embeddings} A \emph{homomorphism} from a graph $H$ to a graph $G$ is a mapping $h: V(H)\rightarrow V(G)$ that preserves adjacencies. In other words, for every edge $\{u,v\}\in E(H)$ it holds that $\{h(u),h(v)\}\in E(G)$. We write $\homs{H}{G}$ for the set of all homomorphisms from $H$ to $G$. A homomorphism inducing a bijection of vertices and satisfying $\{u,v\} \in E(H)$ if and only if $\{f(u),f(v)\} \in E(G)$ is called an \emph{isomorphism} and we say that two graphs $H$ and $\hat{H}$ are \emph{isomorphic} if there exists an isomorphism from $H$ to
$\hat{H}$.
We write $\subs{H}{G}$ and $\indsubs{H}{G}$ for the sets of all subgraphs and induced subgraphs of $G$, respectively, that are isomorphic to $H$.

An isomorphism from a graph to itself is called an \emph{automorphism}. The set of automorphisms of a graph, together with the operation of functional composition constitutes a group, called the \emph{automorphism group} of a graph.
Slightly abusing notation, we will write $\auts{H}$ for both the set of automorphisms of a graph $H$ as well as for the automorphism group of $H$.

An \emph{embedding} is an injective homomorphism and we write $\embs{H}{G}$ for the set of embeddings from $H$ to $G$. If an embedding $h$ from $H$ to $G$ additionally satisfies that $\{h(u),h(v)\} \in E(G)$ implies $\{u,v\}\in E(H)$, we call it a \emph{strong embedding}. We write $\strembs{H}{G}$ for the set of strong embeddings from $H$ to $G$. Observe that the images of embeddings and strong embeddings from $H$ to $G$ are precisely the subgraphs and induced subgraphs of $G$ that are isomorphic to $H$.

\paragraph*{Colored variants} Given graphs $G$ and $H$, we say that $G$ is $H$\emph{-colored} if $G$ comes with a homomorphism $c$ from~$G$ to $H$, called an $H$-\emph{coloring}. Note that, in particular, every edge-subgraph of $H$ can be $H$-colored by the identity function on $V(H)$, which is assumed to be the given coloring whenever we consider $H$-colored edge-subgraphs of $H$ in this paper. Given an edge-subgraph~$F$ of~$H$ and a homomorphism $h$ from $F$ to a $H$-colored graph $G$, we say that $h$ is \emph{color-prescribed} if for all $v\in V(F) = V(H)$ it holds that $c(h(v))= v$. We write $\cphoms{F}{G}$ for the set of all color-prescribed homomorphisms from $F$ to $G$. $\cpstrembs{F}{G}$ is defined similarly for color-prescribed strong embeddings.
We point out that a definition of $\cpemb$ is obsolete as every color-prescribed homomorphism is injective by definition and hence an embedding. Furthermore, we write $\cpsubs{F}{G}$ and $\cpindsubs{F}{G}$ for the sets of images of color-prescribed embeddings and strong embeddings from $F$ to $G$, respectively. Elements of $\cpsubs{F}{G}$ and $\cpindsubs{F}{G}$ are referred to as color-prescribed subgraphs and induced subgraphs.\footnote{The observant reader might have noticed that the sets $\cpsubs{F}{G}$ and $\cphoms{F}{G}$ as well as $\cpindsubs{F}{G}$ and $\cpstrembs{F}{G}$ are essentially the same as a color-prescribed homomorphism is uniquely identified by its image. However, we decided to distinguish those notions in order to make the combinatorial arguments in Section~\ref{sec:main_reductions} more accessible.}

\paragraph*{Graph properties and the alternating enumerator}
A \emph{graph property} is a function $\Phi$ from graphs to $\{0,1\}$ such that for any pair of isomorphic graphs $H$ and $\hat{H}$ we have that $\Phi(H)=\Phi(\hat{H})$. Adapting the notation of Rivest and Vuillemin~\cite{rivestvuillemin}, we define the \emph{alternating enumerator} of a property $\Phi$ and a graph $H$ to be the function
\begin{equation*}
\hat{\chi}(\Phi,H) := \sum_{S\subseteq E(H)} \Phi(H[S]) \cdot (-1)^{\#S}\,.
\end{equation*}
A graph property $\Phi$ is called \emph{edge-monotone} if it is closed under the removal of edges. It is called \emph{monotone} if it is closed under the removal of edges \emph{and} vertices.\footnote{To avoid confusion, we remark that in some literature, e.g. in~\cite{Meeks16} a property is called monotone if it is closed under \emph{addition} of vertices and edges.}
Given a graph property $\Phi$, a positive integer $k$ and a graph $G$, we write $\indsubs{\Phi,k}{G}$ for the set of all induced subgraphs of size $k$ of $G$ that satisfy $\Phi$. Furthermore, given a graph property $\Phi$ and an $H$-colored graph $G$, we write $\cpindsubs{\Phi}{G}$ for the set of all color-prescribed induced subgraphs of size $|V(H)|$ in $G$ that satisfy $\Phi$.

\subsection{Parameterized counting complexity}
The field of parameterized counting was introduced independently by McCartin~\cite{McCartin06} and Flum and Grohe~\cite{FlumG04} and constitutes a hybrid of classical computational counting and parameterized complexity theory. A \emph{parameterized counting problem} is a pair of a function $P:\Sigma^\ast \rightarrow \N$ and a computable parameterization $\kappa: \Sigma^\ast \rightarrow \N$. It is called \emph{fixed-parameter tractable} (FPT) if there exists a computable function $f$ and a deterministic algorithm that computes $P(x)$ in time $f(\kappa(x))\cdot |x|^{O(1)}$ for every $x \in \Sigma^\ast$. A \emph{parameterized Turing reduction} from $(P,\kappa)$ to $(\hat{P},\hat{\kappa})$ is a deterministic FPT algorithm with respect to $\kappa$ that is given oracle access to $\hat{P}$ and that on input $x$ computes $P(x)$ with the additional restriction that there exists a computable function $g$ such that for any oracle query $y$ it holds that $\hat{\kappa}(y) \leq g(\kappa(x))$. We write $(P,\kappa) \fptred (\hat{P},\hat{\kappa})$ if a parameterized Turing reduction exists.

Given a graph $G$ and a positive integer $k$, the parameterized counting problem $\#\clique$ asks to compute the number of complete subgraphs of size $k$ in $G$ and is parameterized by~$k$, that is $\kappa(G,k):=k$. It is complete for the class $\#\W{1}$, which can be seen as a parameterized counting equivalent of $\NP$~\cite{FlumG04}. Evidence for the fixed-parameter intractability of $\#\W{1}$-hard problems is given by the \emph{Exponential Time Hypothesis} (ETH), which asserts that $3$-$\textsc{SAT}$ cannot be solved\footnote{We point out that this includes deterministic and randomized algorithms.} in time $\mathsf{exp}(o(m))$ where $m$ is the number of clauses of the input formula. Assuming ETH, $\#\clique$ cannot be solved in time $f(k)\cdot n^{o(k)}$ for any function $f$~\cite{Chenetal05,Chenetal06} and hence $\#\W{1}$-hard problems are not fixed-parameter tractable.

Given a recursively enumerable class of graphs $\mathcal{H}$, the problem $\#\homsprob(\mathcal{H})$ asks, given a graph $H\in \mathcal{H}$ and an arbitrary graph $G$, to compute $\#\homs{H}{G}$. Its parameterization is given by $\kappa(H,G) := |V(H)|$. The problems $\#\cphomsprob(\mathcal{H})$ and $\#\cpindsubsprob(\mathcal{H})$ are defined similarly.
Furthermore, we define $\#\cpindsubsprob(\Phi)$ to be the problem of, given a graph $G$ that is $H$-colored for some graph $H$, computing $\#\cpindsubs{\Phi}{G}$ and parameterize it by $\kappa(G):=|V(H)|$ --- note that the $H$-coloring of $G$ is part of the input and hence $\kappa$ is well-defined. Finally, the problem $\#\indsubsprob(\Phi)$ asks, given a graph $G$ and a positive integer~$k$, to compute $\#\indsubs{\Phi,k}{G}$ and the parameterization is given by $\kappa(G,k):= k$.

\section{Alternating enumerators and p-edge-transitive graphs}\label{sec:alt_enum}
In this part of the paper we will provide a rough exposition of the work of Rivest and
Vuillemin~\cite{rivestvuillemin} who studied transitive boolean functions to resolve the asymptotic version of Karp's evasiveness conjecture. We will then apply their result to graphs $H$ that are both edge-transitive and have $p^\ell$ many edges for some prime $p$. This will enable us to conclude that the alternating enumerator of $\Phi$ and $H$ is $(\pm 1)$ modulo $p$ whenever $\Phi(H[\emptyset]) \neq \Phi(H)$. We start by introducing some required notions from algebraic graph theory.

\noindent The automorphism group of a graph $H$ induces a group action  on the edges of $H$, given by $h \{u,v\} := \{h(u),h(v)\}$. A group action is \emph{transitive} if there exists only one orbit and a graph $H$ is called \emph{edge-transitive} if the group action on the edges is transitive, that is, if for every pair of edges~$\{u,v\}$ and~$\{\hat{u},\hat{v}\}$ there exists an automorphism $h \in \auts{H}$ such that $h \{u,v\} = \{\hat{u},\hat{v}\}$. If additionally the number of edges of an edge-transitive graph is a prime power $p^\ell$ we call the graph $p$\emph{-edge-transitive}.

\begin{lemma}[Lemma~\ref{lem:alt_enum_intro} restated] \label{lem:alt_enum}
    Let $\Phi$ be a graph property and let $H$ be a $p$-edge-transitive graph such that $\Phi(H[\emptyset]) \neq \Phi(H)$. Then it holds that $\hat{\chi}(\Phi,H) = (\pm 1) \mod p \,.$
\end{lemma}
Lemma~\ref{lem:alt_enum} is implicitly proven in \cite[Theorem 4.3]{rivestvuillemin}, but for completeness we will include a short and self-contained proof, demonstrating a first application of the machinery of Sylow subgroups that we will need later.

For the proofs in this section, let us recall some key results from group theory. Given a prime number $p$, a finite group $\Gamma'$ is called a \emph{$p$-group} if the order $\# \Gamma'$ is a power of $p$. The following is a well-known and central result from the theory of finite groups.
\begin{theorem}[Sylow theorems] \label{thm:Sylow}
 Let $\Gamma$ be a finite group of order $\# \Gamma = p^k m$ for a prime $p$ and an integer $m \geq 1$ coprime to $p$. Then $\Gamma$ contains a subgroup $\Gamma'$ of order $p^k$. Moreover, every other subgroup $\Gamma''$ of $\Gamma$ of order $p^k$ is conjugate to $\Gamma'$, that is there exists $g \in \Gamma$ with $\Gamma'' = g \Gamma' g^{-1}$. In particular, the groups $\Gamma', \Gamma''$ are isomorphic (via the conjugation by $g$).

 Finally, every subgroup $\tilde \Gamma \subseteq \Gamma$ which is a $p$-group is actually contained in some conjugate $g \Gamma' g^{-1}$ of the group $\Gamma'$.
\end{theorem}
A subgroup $\Gamma' \subseteq \Gamma$ as above is called a \emph{$p$-Sylow subgroup} of $\Gamma$.

The following result is a first important application of the Sylow theorems. It can be found as Exercise (E28) in~\cite{lectalgebra}; we include a proof for completeness.
\begin{lemma} \label{lem:ptransitive}
 Let $\Gamma$ be a finite group acting transitively on a set $T$ such that $\# T = p^l$ for some $l \geq 0$. Then the induced action of any $p$-Sylow subgroup $\Gamma' \subseteq \Gamma$ on $T$ is still transitive.
\end{lemma}
\begin{proof}
 Let $t_0 \in T$ be any element, then $T$ is the orbit of $t_0$ under $\Gamma$.
 Let $\operatorname{Stab}_\Gamma(t_0) = \{g \in \Gamma : g t_0=t_0\}$ be the stabilizer of $t_0$ under the action of $\Gamma$. Then by the Orbit-Stabilizer theorem, we have
 \begin{equation} \label{eqn:orbitstab1}
 \# \Gamma = (\# \Gamma t_0) \cdot (\# \operatorname{Stab}_\Gamma(t_0)) = (\# T) \cdot (\# \operatorname{Stab}_\Gamma(t_0)).
 \end{equation}
 As in the Sylow theorems, write $\# \Gamma = p^k m$ with $m$ not divisible by $p$ and let $\Gamma' \subseteq \Gamma$ be a $p$-Sylow subgroup of $\Gamma$, which is of order $p^k$.
 The stabilizer of $t_0$ under the induced action of the subgroup $\Gamma' \subseteq \Gamma$ is given by
 \[\operatorname{Stab}_{\Gamma'}(t_0) = \{g \in \Gamma' : g t_0 = t_0\} = \operatorname{Stab}_\Gamma(t_0) \cap \Gamma'.\]
 Clearly this is a subgroup of the group $\Gamma'$ and by Lagrange's theorem, the order of $\operatorname{Stab}_{\Gamma'}(t_0)$ divides the order $p^k$ of $\Gamma'$. Thus it is itself a power of $p$, say $\#\operatorname{Stab}_{\Gamma'}(t_0) = p^n$.

 On the other hand, $\operatorname{Stab}_{\Gamma'}(t_0)$ is also a subgroup of $\operatorname{Stab}_\Gamma(t_0)$. Inserting the order of $\Gamma$ and the size of $T$ in equation (\ref{eqn:orbitstab1}) we obtain
 \begin{equation} \label{eqn:orbitstab12}
 p^k m = p^l \cdot (\# \operatorname{Stab}_\Gamma(t_0)),
 \end{equation}
 and thus $\# \operatorname{Stab}_\Gamma(t_0)$ can at most contain a factor of $p^{k-l}$. Again, by Lagrange's theorem, the order $p^n$ of the subgroup  $\operatorname{Stab}_{\Gamma'}(t_0)$ divides the order of $\operatorname{Stab}_\Gamma(t_0)$ and thus $n \leq k-l$.

 Finally, by the Orbit-Stabilizer theorem applied to the action of $\Gamma'$ on $t_0$, we have
 \begin{equation} \label{eqn:orbitstab2}
 p^k = \# \Gamma' = (\# \Gamma' t_0) \cdot (\# \operatorname{Stab}_{\Gamma'}(t_0)) = (\# \Gamma' t_0) \cdot p^n.
 \end{equation}
 Thus, on the one hand we obtain $\#\Gamma' t_0 = p^{k-n} \geq p^{k-(k-l)} = p^l$.
 On the other hand we obtain $\Gamma' t_0 \subseteq T$ and thus $\# \Gamma' t_0 \leq \# T = p^l$.
 Hence we have the equality $\# \Gamma' t_0 = p^l = \#T$  and thus $\Gamma' t_0 = T$.
 In other words, the action of $\Gamma'$ on $T$ is transitive, finishing the proof.
\end{proof}

This result allows us to give a short proof of Lemma~\ref{lem:alt_enum} above. 

\begin{proof}[Proof of Lemma \ref{lem:alt_enum}]
 Let $\Gamma = \auts{H}$ be the automorphism group of the graph $H$, then by assumption its action on the set $E(H)$ of edges of $H$ is transitive. By Lemma \ref{lem:ptransitive}, any $p$-Sylow subgroup $\Gamma' \subseteq \Gamma$ still acts transitively on $E(H)$. Now consider the sum
 \begin{equation*}
\hat{\chi}(\Phi,H) = \sum_{S\subseteq E(H)} \Phi(H[S]) \cdot (-1)^{\#S}\,.
\end{equation*}
The action of $\Gamma'$ on $E(H)$ induces an action of
$\Gamma'$ on the set of subsets $\mathcal{P}(E(H)) := \{S \subseteq E(H)\}$ of $E(H)$. Indeed, for $S \subset E(H)$ and $g \in \Gamma'$ we define \mbox{$g S = \{ g s : s \in S\}$}. For this action, the set $\mathcal{P}(E(H))$ can be written as a disjoint
union of the orbits $\Gamma' S_0$ of a set $\mathcal{S}\subseteq \mathcal{P}(E(H))$ of representatives $S_0$. (Recall that for a group action two orbits are either disjoint or equal.) This allows us to write the sum above as
 \begin{equation*}
\hat{\chi}(\Phi,H) = \sum_{S_0 \in \mathcal{S}} \sum_{S \in \Gamma'S_0} \Phi(H[S]) \cdot (-1)^{\#S}\,.
\end{equation*}
 Until now we have just reordered the summands above, combining all summands for $S$ in the same $\Gamma'$ orbit.

 Now since all elements $g \in \Gamma' \subseteq \auts{H}$ act by graph automorphisms on $H$, we have that the graphs $H[g S_0]$ and $H[S_0]$ are isomorphic, so in particular $\Phi(H[gS_0]) = \Phi(H[S_0])$. Applying this to the formula for $\hat{\chi}(\Phi,H)$ above, we get
  \begin{equation}\label{eqn:orbitstab22}
\hat{\chi}(\Phi,H) = \sum_{S_0 \in \mathcal{S}} (\# \Gamma'S_0) \cdot  \Phi(H[S_0]) \cdot (-1)^{\#S_0}\,.
\end{equation}
 Now by the Orbit-Stabilizer theorem, the size $\# \Gamma'S_0$ of the orbit of $S_0$ divides the order $p^k$ of $\Gamma'$, so $\# \Gamma'S_0$ is itself a power of $p$.
 Further, unless $S_0 \subseteq E(H)$ is invariant under $\Gamma'$, the size of its orbit
 $\# \Gamma'S_0$ is a \emph{positive} power of $p$ and thus congruent to $0$ mod $p$.
 However, the only two sets $S_0 \subseteq E(H)$ invariant under $\Gamma'$ are
 $S_0 = \emptyset$ and $S_0 = E(H)$:
 Consider for the sake of contradiction the case where $S_0$ is invariant under $\Gamma'$ and
 nonempty, but not the whole set $E(H)$. Then $S_0$ contains an element $e_0$, and since $S_0$ is
 $\Gamma'$-invariant, $S_0$ also contains the entire orbit of $e_0$ under $\Gamma'$.
 But since $\Gamma'$ acted transitively on $E(H)$, $S_0$  must have been the whole set $E(H)$,
 yielding a contradiction.

 To summarize, when computing $\hat{\chi}(\Phi,H)$ modulo $p$ all but two summands in the sum
 in Equation~(\ref{eqn:orbitstab22}) are congruent to $0$. Hence, we can simplify
 Equation~(\ref{eqn:orbitstab22}) to
 \begin{equation*}
\hat{\chi}(\Phi,H) =  \Phi(H[\emptyset]) + \Phi(H[E(H)]) \cdot (-1)^{\#E(H)} = \Phi(H[\emptyset]) - \Phi(H) \mod p\,.
\end{equation*}
Note that we use the fact that for $p>2$ we have that $\#E(H)$ is odd since it is a prime power
and for $p=2$ we have $-1 = 1$ modulo $p$. Now, the condition $\Phi(H[\emptyset]) \neq \Phi(H)$ exactly gives us $\Phi(H[\emptyset]) - \Phi(H) = \pm 1 \mod p$. 
\end{proof}

There are two main examples for $p$-edge-transitive graphs. The first example is
the class of the complete, bipartite graphs $K_{p^l,p^m}$ with $l,m \geq 0$.
The graph $K_{p^l,p^m}$ has $p^{l+m}$ edges and the automorphism group clearly acts transitively
on the edges of that graph.
The second example is the class of wreath graphs $W_{p^k}$ for $k \geq 1$.
The graph $W_{p^k}$ has $p^k$ vertices that can be decomposed in disjoint sets $V_0, \ldots, V_{p-1}$ of order $p^{k-1}$ each, and edges $\{v_i, v_{i+1}\}$ for each $i=0, \ldots, p-1$ and
vertices $v_i \in V_i, v_{i+1} \in V_{i+1}$ (where it is understood that $V_{p}=V_{0}$).
Thus in total, $W_{p^k}$ has $p^{2k-1}$ edges, except for $p=2$ where it has $2^{2k-2}$ edges.
The graph $W_{p^k}$ can be seen as the lexicographical product of a $p$-cycle with a graph consisting of $p^{k-1}$ disjoint vertices. For $k=1$ we exactly obtain the $p$-cycle.
To see that $W_{p^k}$ is edge-transitive, we observe that on the one hand, for fixed $i$ we can apply an arbitrary permutation on $V_i$ leaving the graph invariant.
On the other hand, there exists a ``rotational action'' sending $V_j$ to $V_{j+1}$ for $j=0, \ldots, p-1$, which also leaves the graph invariant. Using these two types of automorphisms, we can map every edge to every other edge. 

The following result tells us that in a certain sense the graphs $K_{p^l,p^m}$ and $W_{p^k}$ are the maximal $p$-edge-transitive graphs. A graph $G$ is called \emph{vertex-transitive} if its automorphism group $\auts{G}$ acts transitively on its set of vertices $V(G)$.
\begin{theorem}[Theorem~\ref{thm:pedgetransgraphs_intro} restated] \label{thm:pedgetransgraphs}
 Let $G$ be a connected $p$-edge-transitive graph. Then either $G$ is bipartite (and thus a subgraph of a graph of the form $K_{p^l,p^m}$ for some $l,m \geq 0$) or $G$ is vertex-transitive and an edge-subgraph of $W_{p^k}$  for $k \geq 1$ (or both).
\end{theorem}

For the proof of Theorem~\ref{thm:pedgetransgraphs}, we will make use of the following well-known result about the relation between edge and vertex-transitivity \cite[Proposition 15.1]{alggraphthy}.

    \begin{lemma}\label{lem:vtransbipartdichot}
 Let $G$ be a connected graph and let $\Gamma \subseteq \auts{G}$ be a subgroup acting transitively on the set of edges $E(G)$. Then either $\Gamma$ acts transitively on the set of vertices $V(G)$ (and thus $G$ is vertex-transitive) or $G$ is bipartite (or both).
\end{lemma}

The proof from \cite{alggraphthy} carries over verbatim to the setting of the previous lemma, by replacing the full group $\auts{G}$ with the subgroup $\Gamma$. We include it for completeness.
\begin{proof}
 Let $e_0 = \{v_1, v_2\} \in E(G)$ be some edge. Then for each vertex $w_1$ and some edge $\{w_1, w_2\}$ incident to $w_1$ there exists an automorphism $g \in \Gamma$ sending $\{v_1, v_2\}$ to $\{w_1, w_2\}$, thus either $w_1 = g v_1$ or $w_1=g v_2$. This proves that the set $V(G)$ is the union of the orbits of $v_1, v_2$ under $\Gamma$. If these two orbits intersect, then in fact $\Gamma v_1 = \Gamma v_2$ (orbits under a group action are either disjoint or equal), and so $\Gamma v_1 = V(G)$ since the two sets above cover $V(G)$. On the other hand, if the two orbits are disjoints, then they form a partition making $G$ into a bipartite graph. Indeed, any edge of $G$ is in the orbit of the edge $\{v_1, v_2\}$ and thus connects an element of $\Gamma v_1$ to an element of $\Gamma v_2$.
\end{proof}

\begin{proof}[Proof of Theorem \ref{thm:pedgetransgraphs}]
 Let $G$ be a $p$-edge-transitive, non-bipartite graph. Then by Lemma~\ref{lem:ptransitive} any $p$-Sylow subgroup $\Gamma \subseteq \auts{G}$ still acts transitively on the edges $E(G)$ of $G$.
 By Lemma \ref{lem:vtransbipartdichot}, since $G$ is not bipartite, the group $\Gamma$ acts transitively on the set of vertices $V(G)$ (and thus $G$ is also vertex-transitive). 
We observe that in this case, by the Orbit-Stabilizer theorem, we have $\# V(G) = p^k$ for some $k \geq 1$. We claim that then $G$ is a subgraph of $W_{p^k}$.

 To see this, let us reformulate our situation slightly: We identify the vertex set $V(G)$ with the set $[p^k]=\{1, \ldots, p^k\}$.
 Then we can canonically identify $\auts{G}$ as a subgroup of $S_{p^k}$, the symmetric group on $[p^k]$
 (this is because a graph automorphism is uniquely determined by its action on the vertices of a graph).
 Inside $\auts{G}$ we have the subgroup $\Gamma$, which is a $p$-group. By the Sylow theorem, there exists a $p$-Sylow subgroup $\Gamma' \subseteq S_{p^k}$ containing $\Gamma$. Since the action of $\Gamma$ is transitive on the set of edges $E(G)$, we can obtain $E(G)$ by starting with some edge $e_0 = \{v_1, v_2\} \in E(G)$ with $v_1, v_2 \in [p^k]$ and taking its orbit $\{ \{g v_1, g v_2\} : g \in \Gamma\} = E(G)$. But note that by instead taking the orbit of $e_0$ under $\Gamma' \subseteq S_{p^k}$ we get at least this set of edges and maybe more. Denote by $G'$ the graph with vertices $[p^k]$ and edges $\{ \{g v_1, g v_2\} : g \in \Gamma'\}$. We claim that $G' \cong W_{p^k}$.

 To show this we will explicitly identify the $p$-Sylow subgroup $\Gamma' \subseteq S_{p^k}$ (recall that by the Sylow theorem it is unique up to conjugation, that is reordering of the elements of $[p^k]$).

 First note that $S_{p^k}$ has $(p^k)!$ elements. Inductively one sees that the highest power of $p$ appearing in this number is $p^{e(k)}$ for $e(k)=p^{k-1}+p^{k-2} + \ldots + p + 1$. We will inductively construct a subgroup $\Gamma(p,k)$ of $S_{p^k}$ with $p^{e(k)}$ elements, which then is a $p$-Sylow subgroup. We note that a description of such a $p$-Sylow subgroup is given in \cite{sylowpsubgps}. 

 For $k=1$ we have $e(k)=1$ and a $p$-Sylow subgroup $\Gamma(p,1) \subseteq S_p$ is generated by a cyclic permutation $1 \mapsto 2, 2 \mapsto 3, \ldots, p \mapsto 1$ of the elements of $[p]$. The group $\Gamma(p,1)$ is isomorphic to the cyclic group $\mathbb{Z}/p\mathbb{Z}$.

 Now assume we constructed $\Gamma(p,k-1)$ for some $k\geq 2$, then we first note that a product of $p$ copies $\prod_{i=0}^{p-1} \Gamma(p,k-1)$ of $\Gamma(p,k-1)$ acts on $[p^{k}]$ where the $i$-th factor acts by permutations on the elements $ip^{k-1}+1, ip^{k-1}+2, \ldots , ip^{k-1} +p^{k-1}=(i+1)p^{k-1}$. All of these actions commute, so we can see the product $\prod_{i=0}^{p-1} \Gamma(p,k-1)$ as a subgroup of $S_{p^{k}}$. However, there is a further action of $\mathbb{Z}/p\mathbb{Z}$ on $[p^{k}]$ sending $j$ to $j+p^{k-1}$ (modulo $p^{k}$). This action cyclically permutes the $p$ blocks of $p^{k-1}$ elements in $[p^{k}]$ on which the $p$ factors of $\prod_{i=0}^{p-1} \Gamma(p,k-1)$ act. Thus these two actions do not commute, but indeed they induce an action of the semidirect product
 \[\Gamma(p,k) = \left(\prod_{i=0}^{p-1} \Gamma(p,k-1)\right) \rtimes \mathbb{Z}/p\mathbb{Z},\]
 where $\mathbb{Z}/p\mathbb{Z}$ acts on $\prod_{i=0}^{p-1} \Gamma(p,k-1)$ by permuting the factors of the product. We claim that $\Gamma(p,k)$ is the desired $p$-Sylow subgroup of $S_{p^{k}}$. 
 
 \noindent Indeed, as a semidirect product its number of elements is
 \[\# \Gamma(p,k) = (\# \Gamma(p,k-1))^p \cdot p = (p^{e(k-1)})^p \cdot p = p^{pe(k-1)+1}=p^{e(k)},\]
 so it has the correct number of elements and is indeed a subgroup of $S_{p^{k}}$.

 Now recall what we want to show: for a pair $\{v_1, v_2\}$ of vertices forming an edge of our original graph $G$, we want to show that the graph $G'$ with edges $\{ \{g v_1, g v_2\} : g \in \Gamma'\cong \Gamma(p,k)\}$ is isomorphic to the wreath graph $W_{p^k}$. By relabeling the vertices (that is performing a conjugation in $S_{p^k}$) we may assume that $\Gamma'=\Gamma(p,k)$. Furthermore, by a translation in the group $\Gamma(p,k)$, which acts transitively on the elements of $[p^k]$, we may assume that $v_1=1$. Now if $v_2$ were in the first block $[p^{k-1}]$ of vertices, on which the first factor $\Gamma(p,k-1)$ operates, then it is easy to see that the resulting graph $G'$ would not be connected: the first factor $\Gamma(p,k-1)$ would send the edge $\{1,v_2\}$ only to edges within the first block $[p^{k-1}]$ and then the cyclic permutation by the factor $\mathbb{Z}/p\mathbb{Z}$ would send this pattern of edges to the $p-1$ other blocks, giving us a disjoint union of $p$ graphs. This is not possible, since our original graph $G$ is a subgraph of $G'$ and also was assumed to be connected.

 Thus we may assume that $v_2$ is in one of the other blocks $P(a)=[p^{k-1}]+ip^{k-1}$ for $a=1, \ldots, p-1$. Now we want to argue that we can reorder these blocks, sending $P(a)$ to $P(1)$ and leaving $P(0)$ invariant, such that the group action of $\Gamma(p,k)$ is respected. And indeed, let $b \in \mathbb{Z}/p\mathbb{Z}$ be the multiplicative inverse of $a$ (such that $ab = 1 \mod p$), then there is a permutation of $[p^k]$ sending the block $P(i)$ to $P(i\cdot b \mod p)$ (where the block is
 just translated as a whole, not permuting the elements inside). And indeed, we see that $P(a)$ is sent to $P(1)$. The reason why this permutation respects the form of the action\footnote{To be precise, what happens is the following: the map sending $P(i)$ to $P(i\cdot b \mod p)$ is a permutation of $[p^k]$, that is an element $\sigma \in S_{p^k}$. What we are claiming is that the subgroup $\Gamma(p,k) \subseteq S_{p^k}$ is stable under the conjugation by $\sigma$, that
 is $\Gamma(p,k)=\sigma \Gamma(p,k) \sigma^{-1}$. So $\sigma$ is a relabeling of the vertices of our graph $G'$ which leaves the graph itself invariant.} of $\Gamma(p,k)$ is that multiplication by $b$ induces a group isomorphism $\mathbb{Z}/p\mathbb{Z} \to \mathbb{Z}/p\mathbb{Z}$ on the semidirect factor $\mathbb{Z}/p\mathbb{Z}$ of $\Gamma(p,k)$.

 To summarize, we can assume without loss of generality that we start with an edge $\{1,v_2\}$ with $v_2$ in the second block of vertices. But then it is easy to see that the graph $G'$ obtained by taking the orbit of $\{1,v_2\}$ under $\Gamma(p,k)$ is indeed the wreath graph $W_{p^k}$. Indeed, the group $\Gamma(p,k)$ acts transitively within each of the $p$ blocks of vertices (since the $i$-th factor $\Gamma(p,k-1)$ above acts transitively there), so every edge from the first to the second block is in the orbit of $\{1,v_2\}$. Then finally the cyclic permutation action of $\mathbb{Z}/p\mathbb{Z}$ sends these edges to the set of all edges between blocks $i$ and $i+1$, which exactly gives the set of edges of the wreath graph. This finishes the proof.
\end{proof}

\section{The main reduction: From homomorphisms to induced subgraphs}\label{sec:main_reductions}
\noindent In what follows we will construct a sequence of reductions, starting from $\#\homsprob(\mathcal{H})$ and ending in $\#\indsubsprob(\Phi)$. Here, $\mathcal{H}$ is a recursively enumerable set of $p$-edge-transitive graphs and $\Phi$ is a graph property such that for every graph $H \in \mathcal{H}$ we have that $\Phi(H[\emptyset]) \neq \Phi(H)$.

\noindent More precisely, we will prove that
\begin{equation}\label{eq:red_sequence}
\#\homsprob(\mathcal{H}) \stackrel{\text{Lemma~\ref{lem:hom_to_colhom}}}{\fptred} \#\cphomsprob(\mathcal{H}) \stackrel{\text{Lemma~\ref{lem:main_lem}}}{\fptred} \#\cpindsubsprob(\Phi) \stackrel{\text{Lemma~\ref{lem:col_ind_sub_to_ind_sub}}}{\fptred} \#\indsubsprob(\Phi)
\end{equation}
In particular, all of those reductions will be tight in the sense that conditional lower bounds on the fine-grained complexity of $\#\homsprob(\mathcal{H})$ immediately transfer to $\#\indsubsprob(\Phi)$. For the hardness results we rely on a result of Dalmau and Jonsson~\cite{DalmauJ04} stating that the problem $\#\homsprob(\mathcal{H})$ is known to be $\#\W{1}$-hard whenever $\mathcal{H}$ is recursively enumerable and of unbounded treewidth.\footnote{We remark that the graph parameter of
treewidth is not used explicitly in this work. Hence we omit the definition and refer the interested reader e.g. to Chapter~11 in~\cite{FlumG06}.} Here a class of graphs is said to have unbounded treewidth if for every $b \in \N$ there exists a graph in the class with treewidth at least $b$. 

\paragraph*{Reducing homomorphisms to color-prescribed homomorphisms}
In the first reduction we are given graphs $H$ and $G$ and the goal is to compute $\#\homs{H}{G}$ using an oracle for $\#\cphoms{H}{\star}$. This can be done by taking precisely $|V(H)|$ copies of the vertices of $G$, that is, one for each vertex in $H$ and then adding an edge between two vertices $u$ and $v$ if they have been adjacent in $G$ and the vertices of $H$ corresponding to the copies of $V(G)$ that contain $u$ and $v$ are adjacent in $H$ as well. The construction is formalized in the proof of the following lemma. In particular it is shown that the resulting graph $\hat{G}$ is $H$-colored.

\begin{lemma}\label{lem:hom_to_colhom}
Let $H$ be a graph. There exists an algorithm $\mathbb{A}$ that is given a graph $G$ as input and has oracle access to the function $\#\cphoms{H}{\star}$ and computes $\#\homs{H}{G}$ in time $f(|V(H)|)\cdot |V(G)|$ where $f$ is a computable function. Furthermore, every oracle query~$\hat{G}$ satisfies $|V(\hat{G})| \leq f(|V(H)|) \cdot |V(G)|$.
\end{lemma}
\begin{proof}
Let $k = |V(H)|$. It will be convenient to assume that $V(H)=[k]$. Given $G$, we construct a graph $\hat{G}$ as follows. The vertex set of $\hat{G}$ is defined to be
\[V(\hat{G}) = \bigcup_{i=1}^{k} V_i \,,\]
where $V_i = \{v_i ~|~ v \in V(G)\}$ is a copy of $V(G)$ identified with vertex $i \in V(H)$. We add an edge $\{u_i,v_j\}$ to $\hat{G}$ if and only if $\{i,j\} \in E(H)$ and $\{u,v\} \in E(G)$. Now it can easily be verified that the function $c: V(\hat{G}) \to V(H)$ given by $c(v_i) := i$ is an $H$-coloring of $\hat{G}$. Furthermore it is easy to see that \[\#\cphoms{H}{\hat{G}} = \#\homs{H}{G}\,,\]
which concludes the proof.
\end{proof}

\paragraph*{Reducing color-prescribed homomorphisms to color-prescribed induced subgraphs}
%
The reduction from color-prescribed homomorphisms to color-prescribed induced subgraphs requires the introduction of an $H$-colored variant of the framework of graph motif parameters, which was explicitly introduced in~\cite{CurticapeanDM17} and implicitly used in~\cite{ChenM16}.
More precisely, given an $H$-colored graph $G$ and a property $\Phi$, we will express $\#\cpindsubs{\Phi}{G}$ as a linear combination of color-prescribed homomorphisms, that is, terms of the form $\#\cphoms{H[S]}{G}$.
In a first step, we show complexity monotonicity for linear combinations of color-prescribed homomorphisms. While this property allows a quite simple proof, a second step, in which we study the coefficient of $\#\cphoms{H}{G}$ requires a thorough understanding of the alternating enumerator of $\Phi$ and $H$.
In case of $p$-edge-transitive graphs, the latter is provided by Lemma~\ref{lem:alt_enum}.

\noindent We start by introducing a colored variant of the tensor product of graphs (see e.g. Chapter~5.4.2 in~\cite{Lovasz12}). Given two $H$-colored graphs $G$ and $\hat{G}$ with colorings $c$ and $\hat{c}$ we define their \emph{color-prescribed} tensor product $G\timesc \hat{G}$ as the graph with vertices $V=\{(v,\hat{v})\in V(G)\times V(\hat{G}) ~|~c(v)=\hat{c}(\hat{v})\}$
and edges between two vertices $(v,\hat{v})$ and $(u,\hat{u})$ if and only if $\{v,u\} \in E(G)$ and $\{\hat{v},\hat{u}\} \in E(\hat{G})$.
The next lemma states that $\#\cphom$ is linear with respect to $\timesc$. 

\begin{lemma}\label{lem:homomorphic}
Let $H$ be a graph, let $F$ be an edge-subgraph of $H$ and let $G$ and $\hat{G}$ be $H$-colored. Then we have that
\begin{equation*}
\#\cphoms{F}{G \timesc \hat{G}} = \#\cphoms{F}{G} \cdot \#\cphoms{F}{\hat{G}}\,.
\end{equation*}
\end{lemma}
\begin{proof}
It can easily be verified that the function $b(h,\hat{h})(v) := (h(v),\hat{h}(v))$ that assigns elements in $\cphoms{F}{G} \times \cphoms{F}{\hat{G}}$ to elements in $\cphoms{F}{G \timesc \hat{G}}$ is a well-defined bijection.
\end{proof}

The proof of the complexity monotonicity property for color-prescribed homomorphisms (Lemma~\ref{lem:color_prescribed_monotonicity}) will require to solve a system of linear equations. The following lemma proves that the corresponding matrix is non-singular.
\begin{lemma}\label{lem:invertible}
Let $H$ be a graph and let $M$ be a quadratic matrix of size $2^{|E(H)|}$ such that the rows and columns are identified by the subsets of edges of $H$. Furthermore assume that the entries of $M$ are given by $M(S,T) :=  \#\cphoms{H[S]}{H[T]}$. Then $M$ is non-singular. This holds true even if $M$ is considered as a matrix over $\Z_p$, that is, the field with $p$ elements. In the latter case, the entries are taken modulo $p$.
\end{lemma}
\begin{proof}
We fix any linear extension $\lesssim$ of the subset inclusion relation on $E(H)$ and order the columns and rows of $M$ accordingly. We claim that $M$ is triangular. To see this we first observe that $M(S,S)= 1$ for every $S$, given by the identity homomorphism from $H[S]$ to $H[S]$ which is, of course, color-prescribed. Now consider $M(S,T)$ for some $T\neq S$ with $T \lesssim S$. It follows that there exists an edge $\{u,v\}$ in $S\setminus T$ since $\lesssim$ linearly extends subset inclusion. Now assume that there exists a color-prescribed homomorphism $h$ from~$H[S]$ to~$H[T]$. By color-prescribedness we have that $h(u)=u$ and $h(v)=v$, contradicting the fact that $h$ is a homomorphism and $\{u,v\} \notin T$. Hence $M(S,T) = 0$ and, consequently, $M$ is triangular.
\end{proof}

\noindent We are now prepared to prove the color-prescribed variant of complexity monotonicity.

\begin{lemma}[Complexity monotonicity]\label{lem:color_prescribed_monotonicity}
Let $H$ be a graph and let $a$ be a function from edge-subgraphs of $H$ to rationals. There exists an algorithm $\mathbb{A}$ that is given an $H$-colored graph $G$ as input and has oracle access to the function
\[\sum_{S \subseteq E(H)} a(H[S]) \cdot \#\cphoms{H[S]}{ \star}\,,\]
and computes $\#\cphoms{H[S]}{G}$ for all $S$ such that $a(H[S])\neq 0$ in time $f(|H|)\cdot |V(G)|$ where $f$ is a computable function. Furthermore, every oracle query $\hat{G}$ satisfies $|V(\hat{G})| \leq f(|H|)\cdot |V(G)|$.
\end{lemma}
\begin{proof}
Using Lemma~\ref{lem:homomorphic} we have that for every $H$-colored graph $F$ it holds that
\begin{align}
~&\sum_{S \subseteq E(H)} a(H[S]) \cdot \#\cphoms{H[S]}{(G\timesc F)} \\ =&\sum_{S \subseteq E(H)} a(H[S]) \cdot \#\cphoms{H[S]}{G}\cdot \#\cphoms{H[S]}{F}\,,
\end{align}
which we can evaluate for $F=H[\emptyset],\dots,H[E(H)]$. This induces a system of linear equations and the corresponding matrix is non-singular by Lemma~\ref{lem:invertible}. Consequently, the numbers $a(H[S]) \cdot \#\cphoms{H[S]}{G}$ are uniquely determined and can be computed by solving the system using Gaussian elimination. Finally, we obtain the numbers $\#\cphoms{H[S]}{G}$ by multiplying with $a(H[S])^{-1}$ whenever $a(H[S]) \neq 0$.
\end{proof}

It remains to express the number of color-prescribed induced subgraphs that satisfy a property $\Phi$ as a linear combination of color-prescribed homomorphisms.

\begin{lemma}\label{lem:cpindprops_to_cphoms}
Let $H$ be a graph, let $\Phi$ be a graph property and let $G$ be an $H$-colored graph. Then it holds that
\begin{equation*}
\#\cpindsubs{\Phi}{G} = \sum_{S \in E(H)} \Phi(H[S])  \sum_{J \subseteq E(H)\setminus S} (-1)^{\#J} \cdot \#\cphoms{[H[S \cup J]}{G} \,.
\end{equation*}
Moreover, the absolute values of the coefficient of $\#\cphoms{H}{G}$ and
$\hat{\chi}(\Phi,H)$ are equal.
\end{lemma}
\begin{proof}
We start by establishing the following claim.
\begin{claim}\label{clm:cpindsubs_to_cpsubs}
Let $H$ be graph, let $S\subseteq E(H)$ and let $G$ be an $H$-colored graph. Then we have that
\begin{equation*}
\#\cpindsubs{H[S]}{G} = \sum_{J\subseteq E(H)\setminus S} (-1)^{\#J} \cdot \#\cpsubs{H[S\cup J]}{G}\,.
\end{equation*}
\end{claim}
\begin{claimproof}
It holds that
\begin{equation}
\cpindsubs{H[S]}{G} = \cpsubs{H[S]}{G} \setminus \left(\bigcup_{e \in E(H)\setminus S} \cpsubs{H[S\cup\{e\}]}{G} \right)\,,
\end{equation}
and hence, by inclusion-exclusion,
\begin{align}
~&\#\cpindsubs{H[S]}{G})\\
=~& \#\cpsubs{H[S]}{G} - \sum_{\emptyset \subsetneq J \subseteq E(H)\setminus S} (-1)^{\#J-1}\cdot  \#\cpsubs{H[S\cup J]}{G} \\
=~&\sum_{J\subseteq E(H)\setminus S} (-1)^{\#J} \cdot \#\cpsubs{H[S\cup J]}{G}\,.
\end{align}
\end{claimproof}
Now we have that
\begin{align}
\#\cpindsubs{\Phi}{G} =& \sum_{S \in E(H)} \Phi(H[S]) \cdot \#\cpindsubs{H[S]}{G}\label{eq:t1}\\
=& \sum_{S \in E(H)} \Phi(H[S])  \sum_{J\subseteq E(H)\setminus S} (-1)^{\#J} \cdot \#\cpsubs{H[S\cup J]}{G}\label{eq:t2}\\
=& \sum_{S \in E(H)} \Phi(H[S]) \sum_{J\subseteq E(H)\setminus S} (-1)^{\#J} \cdot \#\cphoms{H[S\cup J]}{G}\label{eq:t3}
\end{align}
where (\ref{eq:t1}) follows from the definition of $\cpindsubs{\Phi}{G}$, (\ref{eq:t2}) is Claim~\ref{clm:cpindsubs_to_cpsubs} and (\ref{eq:t3}) holds as color-prescribed homomorphisms are injective and a color-prescribed embedding is uniquely identified by its image. Collecting for the coefficient of $\#\cphoms{H}{G}$ yields
\begin{align}
\sum_{S \in E(H)} \Phi(H[S])\cdot (-1)^{\#E(H)-\#S} = (-1)^{\#E(H)}\cdot \hat{\chi}(\Phi,H) \,.
\end{align}
\end{proof}
The application of the complexity monotonicity property for color-prescribed homomorphisms (Lemma~\ref{lem:color_prescribed_monotonicity}) requires non-zero coefficients. However, this can be guaranteed for the coefficient of interest in case of $p$-edge-transitive graphs as shown in Section~\ref{sec:alt_enum}. Formally, the reduction is constructed as follows.

\begin{lemma}\label{lem:main_lem}
Let $\Phi$ be a graph property and let $H$ be a $p$-edge-transitive graph such that $\Phi(H[\emptyset]) \neq \Phi(H)$. There exists an algorithm $\mathbb{A}$ that is given an $H$-colored graph $G$ as input and has oracle access to the function $\#\cpindsubs{\Phi}{\star}$ and computes $\#\cphoms{H}{G}$ in time $f(|H|)\cdot |V(G)|$ where $f$ is a computable function. Furthermore, every oracle query $\hat{G}$ is $H$-colored as well and satisfies $|V(\hat{G})| \leq f(|H|)\cdot |V(G)|$.
\end{lemma}
\begin{proof}
Using Lemma~\ref{lem:cpindprops_to_cphoms} we can express $\#\cpindsubs{\Phi}{\star}$ as a linear combination of color-prescribed homomorphisms. In particular, the coefficient of $\#\cphoms{H}{\star}$ is $(\pm 1) \cdot \hat{\chi}(\Phi,H)$ and by Lemma~\ref{lem:alt_enum} we have that this number is non-zero whenever $H$ is $p$-edge-transitive and $\Phi(H[\emptyset]) \neq \Phi(H)$. Hence we can use the algorithm from Lemma~\ref{lem:color_prescribed_monotonicity} to compute $\#\cphoms{H}{G}$ in the desired running time.
\end{proof}

\paragraph*{Reducing color-prescribed induced subgraphs to uncolored induced subgraphs}
The last part of the reduction sequence allows us to get rid of the colors. More precisely, we will reduce the problem of counting color-prescribed induced subgraphs of an $H$-colored graph to the problem of counting uncolored induced subgraphs of size $|V(H)|$ in a graph, both with respect to some property $\Phi$. The proof is a straightforward application of the inclusion-exclusion principle.

\begin{lemma}\label{lem:col_ind_sub_to_ind_sub}
Let $\Phi$ be a graph property and let $H$ be a graph with $k$ vertices. There exists an algorithm $\mathbb{A}$ that is given an $H$-colored graph $G$ as input and has oracle access to the function $\#\indsubs{\Phi,k}{\star}$ and computes $\#\cpindsubs{\Phi}{G}$ in time $f(k)\cdot |V(G)|$ where $f$ is a computable function. Furthermore, every oracle query $\hat{G}$ satisfies $|V(\hat{G})| \leq |V(G)|$ and, in particular, $\hat{G}$ allows an $H$-coloring as well.
\end{lemma}
\begin{proof}
It will be convenient to assume that $V(H)=[k]$. We first check whether the $H$-coloring $c$ of $G$ is surjective. If this is not the case then there exists some vertex $i\in V(H)$ such that $i \notin \mathsf{im}(c)$ and hence there is no color-prescribed induced subgraph of $G$, so $\mathbb{A}$ can just output $0$. Otherwise, the $H$-coloring of $G$ induces a partition of $V(G)$ in $k$ many non-empty and pairwise disjoint subsets, each associated with some ``color''~$i \in V(H)$. This allows us to equivalently express $\cpindsubs{\Phi}{G}$ in terms of vertex-colorful induced subgraphs:
\begin{equation}
\cpindsubs{\Phi}{G} = \left\lbrace S \subseteq \binom{V(G)}{k}~\middle|~c(S)=[k] ~\wedge~ \Phi(G[S]) = 1\right\rbrace
\end{equation}
By the principle of inclusion and exclusion we obtain that
\begin{equation}
\#\cpindsubs{\Phi}{G}=\sum_{J\subseteq [k]} (-1)^{\#J}\cdot \#\indsubs{\Phi,k}{G_J}\,,
\end{equation}
where $G_J$ is the graph obtained from $G$ by deleting all vertices that are colored with some color in $J$. Hence we can compute $\#\cpindsubs{\Phi}{G}$ using $2^k$ oracle calls. Finally, we observe that $H$-colored graphs are closed under the removal of vertices and therefore every oracle query $G_J$ allows an $H$-coloring.
\end{proof}

\section{Non-trivial monotone properties on bipartite graphs}\label{sec:apply_bipartite}
In the last part of the paper, we apply the algebraic approach which was laid out in the preceding sections to bipartite graph properties. This will allow us to prove our main result. To this end, we say that a set $\mathcal{K}\subseteq \N$ is \emph{dense} if there exists a constant $c$ such that for every $k'\in \N$ there exists $k \in \mathcal{K}$ such that $k' \leq k \leq ck'$. Furthermore, we write $\mathsf{IS}_k$ for the graph with $k$ isolated vertices. The following theorem is obtained by invoking the reduction sequence~(\ref{eq:red_sequence}) to complete bipartite graphs $K_{t,t}$ for prime powers $t=p^k$, which are $p$-edge-transitive (see Section~\ref{sec:alt_enum}). The extension to modular counting can be found in the appendix.
\begin{theorem}[Theorem~\ref{thm:intro_new} restated]\label{thm:main_result}
Let $\Phi$ be a computable graph property and let $\mathcal{K}$ be the set of all prime powers $t$ such that $\Phi(\mathsf{IS}_{2t}) \neq \Phi(K_{t,t})$. If $\mathcal{K}$ is infinite then $\#\indsubsprob(\Phi)$ is $\#\W{1}$ hard. If additionally $\mathcal{K}$ is dense then it cannot be solved in time $f(k)\cdot n^{o(k)}$ for any computable function $f$ unless ETH fails. This holds true even if the input graphs to $\#\indsubsprob(\Phi)$ are restricted to be bipartite.
\end{theorem}
While $\#\W{1}$-hardness will follow by the classification of Dalmau and Jonsson~\cite{DalmauJ04}, hardness under ETH requires a tight reduction from counting cliques, which we will present first. In particular we use a trick inspired by Lemma~1.11 in~\cite{Curticapean15} to make the reduction parsimonious which is required for the extension to modular counting in the subsequent section.

\begin{lemma}\label{lem:parsi_biclique}
There exists an algorithm that, given a positive integer $\ell>1$ and a graph $G$ with $n$ vertices, computes in time $O(\ell n)$ a $K_{\ell,\ell}$-colored graph $G'$ with at most $O(\ell n)$ vertices such that the number of cliques of size $\ell$ in $G$ equals $\#\cphoms{K_{\ell,\ell}}{G'}$.
\end{lemma}
\begin{proof}
 Let the vertex set of $G$ be $\{v_i \mid 1 \leq i \leq n\}$ and let that of $K_{\ell, \ell}$ be $\{a_i, b_i \mid 1 \leq i \leq \ell\}$.
    We now construct the graph $G'$ on the vertex set $\{u_{i,j}, w_{i,j} \mid 1 \leq i \leq \ell, 1 \leq j \leq n\}$ with a $K_{\ell, \ell}$-coloring given by $c(u_{i,j}) = a_i$ and $c(w_{i, j}) = b_i$. We add an edge between $u_{i,j}$ and $w_{i',j'}$ if and only if
    \begin{itemize}
        \item either $(i, j) = (i', j')$,
        \item or $i < i'$, $j < j'$ and the vertices $v_j$ and $v_{j'}$ are adjacent,
        \item or $i > i'$, $j > j'$ and the vertices $v_j$ and $v_{j'}$ are adjacent.
    \end{itemize}
    Let $\{v_{i_1}, \ldots v_{i_{\ell}}\}$ be an $\ell$-clique in $G$.
    Assume w.l.o.g.\ that $i_k < i_{k'}$ for $k < k'$.
    Then the set $\{u_{1,j_1}, \ldots u_{\ell, j_\ell}, w_{1,i_1}, \ldots w_{\ell, j_\ell}\}$ forms a colorful biclique in $G'$, so it gives rise to a color-prescribed homomorphism $h\in \cphoms{K_{\ell, \ell}}{G'}$.
    Now let $h'\in  \cphoms{K_{\ell, \ell}}{G'}$ be a color-prescribed homomorphism.
    Then there has to be the following colorful biclique in $G'$: \[\{u_{1,\alpha_1}, \ldots u_{\ell, \alpha_\ell}, w_{1,\beta_1}, \ldots w_{\ell, \beta_\ell}\}\,.\]
    We first see that for every $i$ we have $\alpha_i = \beta_i$ since there has to be an edge between $u_{i, \alpha_i}$ and $w_{i, \beta_i}$.
    Furthermore the edges enforce $\alpha_j < \beta_{j'}$ for every $j < j'$.
    Thus $\{v_{\alpha_1}, \ldots, v_{\alpha_\ell}\}$ is a clique of size $\ell$ in $G$.
    Since every homomorphism yields $\beta_j = \alpha_j < \beta_{j'} = \alpha_{j'}$ for $j < j'$ there is exactly one homomorphism in $\cphoms{K_{\ell,\ell}}{G'}$ corresponding to each $\ell$-clique in~$G$.
\end{proof}

\begin{proof}[Proof of Theorem~\ref{thm:main_result}]
Let $\Phi$ and $\mathcal{K}$ be as given in Theorem~\ref{thm:main_result}. We define a class of graphs $\mathcal{H}$ as follows:
\[\mathcal{H} = \{K_{t,t}~|~t \in \mathcal{K}\} \,.\]
By the reductions sequence~(\ref{eq:red_sequence}), given by Lemma~\ref{lem:hom_to_colhom}, Lemma~\ref{lem:main_lem} and Lemma~\ref{lem:col_ind_sub_to_ind_sub}, we obtain that $\#\homsprob(\mathcal{H})\fptred \#\indsubsprob(\Phi)$. As $\Phi$ is computable, $\mathcal{H}$ is recursively enumerable. Furthermore, as $\mathcal{K}$ is infinite, we have that there are arbitrary large bicliques in $\mathcal{H}$ and, in particular, the treewidth of $\mathcal{H}$ is unbounded. Therefore $\#\homsprob(\mathcal{H})$, and hence $\#\indsubsprob(\Phi)$, are $\#\W{1}$-hard by the classification of counting homomorphisms due to Dalmau and Jonsson~\cite{DalmauJ04}. For the tight bound under ETH, we reduce from the \emph{decision problem} $\clique$ which asks, given $G$ and $k$, to \emph{decide} whether $G$ contains a clique of size $k$ and which cannot be solved in time $f(k)\cdot n^{o(k)}$ for any computable function $f$, unless ETH fails~\cite{Chenetal05,Chenetal06}. Now assume that~$\mathcal{K}$ is dense and let $(G,k)$ be an instance of $\clique$. By density of $\mathcal{K}$, there exists $\ell \in \mathcal{K}$ such that $k \leq \ell \leq ck$ for some overall constant $c$ independent of $k$. We construct the graph $\hat{G}$ from $G$ by adding $\ell-k$ further vertices and adding edges between all new vertices as well as between every pair of an old and a new vertex. It can then easily be verified that $G$ contains a clique of size $k$ if and only if $\hat{G}$ contains a clique of size $\ell$.

Next we apply Lemma~\ref{lem:parsi_biclique} to $\hat{G}$ and $\ell$, and obtain an $K_{\ell,\ell}$-colored graph $G'$ satisfying that the number of $\ell$-cliques in $\hat{G}$ is equal to $\#\cphoms{K_{\ell,\ell}}{G'}$. Finally, we invoke Lemma~\ref{lem:main_lem} and Lemma~\ref{lem:col_ind_sub_to_ind_sub} to conclude the reduction. In particular, all reductions are tight in the sense that every oracle call for $\#\indsubsprob(\Phi)$ in the final part of the reduction is a pair $(\tilde{G},2\ell)$ where the number of vertices of $\tilde{G}$ is bounded by $O(\ell\cdot|V(G)|)$. As $\ell \leq ck$ we conclude that every algorithm that solves $\#\indsubsprob(\Phi)$ in time $f(k)\cdot n^{o(k)}$ can be used to solve $\clique$ in time $f(k)\cdot n^{o(k)}$ --- just check in the end whether the output is a number greater than zero.

Finally, we point out that for both ($\#\W{1}$ and ETH) hardness results, the last part of the reduction, that is, Lemma~\ref{lem:col_ind_sub_to_ind_sub} only queries for graphs that are $K_{t,t}$-colorable and hence bipartite.
\end{proof}

Note that, in case $\Phi$ or its complement is edge-monotone, we only have to find infinitely many prime powers~$t$ for which $\Phi$ is neither true nor false on the set of all edge-subgraphs of~$K_{t,t}$, which is the case for all sensible, non-trivial properties that do not rely on the number of vertices in some way. If $\Phi$ (or its complement) is monotone, that is, not only closed under the removal of edges, but also under the removal of vertices, then such artificial properties do not exist and we can state the result more clearly as follows.

\begin{corollary}[Theorem~\ref{thm:intro_main_bip} restated]
Let $\Phi$ be a computable monotone graph property such that $\Phi$ and $\neg \Phi$ hold on infinitely many bipartite graphs. Then $\#\indsubsprob(\Phi)$ is $\#\W{1}$-hard and cannot be solved in time $f(k)\cdot n^{o(k)}$ for any computable function $f$ unless ETH fails. This holds true even if the input graphs to $\#\indsubsprob(\Phi)$ are restricted to be bipartite.
\end{corollary}
\begin{proof}
If $\Phi$ is monotone and $\Phi$ and $\neg\Phi$ hold on infinitely many bipartite graphs, then $\Phi(\mathsf{IS}_k)=1$ for all positive integers $k$ and $\Phi(K_{t,t})=0$ for all but finitely many $t$. Hence we can apply Theorem~\ref{thm:main_result} and, in particular, the set $\mathcal{K}$ will contain all but finitely many prime powers and is therefore dense.
\end{proof}

\paragraph*{Conclusion}
We have established hardness for $\#\indsubsprob(\Phi)$ for \emph{any} (edge-)monotone property $\Phi$ that is non-trivial on bipartite graphs. In particular, this holds true even if we count modulo a prime and restrict the input graphs to be bipartite as well. Hence, we did not only significantly extend the set of graph properties $\Phi$ for which the (parameterized) complexity of $\#\indsubsprob(\Phi)$ is understood, but we also generalized many of the prior results, such as~\cite{JerrumM15},~\cite{Meeks16} and parts of~\cite{RothS18} to the cases of bipartite input graphs and modular counting.

As a next step towards a proof of Conjecture~\ref{con:main_conjecture}, we suggest the study of properties that are defined by forbidden induced subgraphs, for which the complexity of $\#\indsubsprob(\Phi)$ is only partially resolved at this point.

\bibliography{references}

\newpage

\appendix

\section{Extension to modular counting}
In this very final part of the paper, we show that our main result (Theorem~\ref{thm:main_result}) can easily be extended to counting modulo a fixed prime:
\begin{theorem}\label{thm:main_result_mod}
Let $p$ be a prime number, let $\Phi$ be a computable graph property and let~$\mathcal{K}$ be the set of all prime powers $t=p^k$ such that $\Phi(\mathsf{IS}_{2t}) \neq \Phi(K_{t,t})$. If $\mathcal{K}$ is infinite then $\mathsf{Mod}_p\indsubsprob(\Phi)$ is $\mathsf{Mod}_p\W{1}$ hard. If additionally $\mathcal{K}$ is dense then it cannot be solved in time $f(k)\cdot n^{o(k)}$ for any computable function $f$ unless ETH fails. This holds true even if the input graphs to $\mathsf{Mod}_p\indsubsprob(\Phi)$ are restricted to be bipartite.
\end{theorem}
Here $\mathsf{Mod}_p\indsubsprob(\Phi)$ asks, given $G$ and $k$, to compute the number of induced subgraphs with $k$ vertices in $G$ that satisfy $\Phi$ \emph{modulo} $p$. The parameterized complexity class $\mathsf{Mod}_p\W{1}$ is defined by the problem of, given $G$ and $k$, deciding whether the number of $k$-cliques in $G$ is $0$ modulo $p$, which is complete for the class (see~\cite{BjorklundDH15} for $p=2$ and Chapter~1.2.2 in~\cite{Curticapean15} for the general case).

First of all, we point out that the modular counting version of Theorem~\ref{thm:intro_main_bip} follows as corollary from the above theorem in the same way Theorem~\ref{thm:intro_main_bip} follows from Theorem~\ref{thm:main_result}. For the proof of Theorem~\ref{thm:main_result_mod} we rely on the following fact stating that all required reductions in Section~\ref{sec:main_reductions} work as well in the case of counting modular a prime number.

\begin{fact}\label{fac:mod_fac}
Let $p$ be a fixed prime. Then Lemma~\ref{lem:main_lem} and Lemma~\ref{lem:col_ind_sub_to_ind_sub} remain true when counting is done modulo $p$ if  the graph $H$ is restricted to be $K_{t,t}$ for some prime power~$t=p^k$.
\end{fact}
The only two non-trivial observations required to verify Fact~\ref{fac:mod_fac} are, first, that $\hat{\chi}(\Phi,K_{t,t}) \neq 0 \mod p$ whenever $\Phi(K_{t,t}[\emptyset]) \neq \Phi(K_{t,t})$ (Lemma~\ref{lem:alt_enum}) and, second, that complexity monotonicity (Lemma~\ref{lem:color_prescribed_monotonicity}) holds for computation modulo $p$ as well, since non-singularity of the matrix~$M$ in the proof is given by Lemma~\ref{lem:invertible} even in case the entries of $M$ are considered to be elements of~$\mathbb{Z}_p$.
The last ingredient for the proof of Theorem~\ref{thm:main_result_mod}, in particular for hardness under ETH, requires a method of isolating cliques that works in the parameterized setting. This is given by the following result of Williams et al.
\begin{lemma}[Lemma~2.1 in~\cite{WilliamsWWY15}]\label{lem:isolation}
Let $p \geq 2$ be an integer, $G,H$ be undirected graphs. Let~$G'$ be a random induced subgraph of $G$ such that each vertex is taken with probability $1/2$,  independently. If there is at least one induced $H$ in $G$, the number of induced $H$ in $G'$ is not a multiple of $p$ with probability at least $2^{-|H|}$.
\end{lemma}

\begin{proof}[Proof of Theorem~\ref{thm:main_result_mod}]
The proof is most similar to the proof of the tight lower bound under ETH in Theorem~\ref{thm:main_result}. We start our reduction from the problem of \emph{finding} a clique of size~$k$. In case $\mathcal{K}$ is dense and we aim to establish the ETH hardness result, we perform the following two initial steps before the main reduction:
\begin{enumerate}
\item Given $G$ and $k$, we construct a graph $\hat{G}$ such that $G$ contains a clique of size $k$ if and only if $\hat{G}$ contains a clique of size $\ell$ where $k\leq \ell\leq ck$ for some overall constant $c$. The details of the construction are given in the proof of Theorem~\ref{thm:main_result}.
\item We use Lemma~\ref{lem:isolation} to isolate an $\ell$-clique in $\hat{G}$, assuming there is any, with high probability.
\end{enumerate}
For the main part of the reduction we then first apply the reduction from counting cliques to counting color-prescribed homomorphisms from the biclique as given by Lemma~\ref{lem:parsi_biclique}. In particular, this reduction is parsimonious. Finally, we proceed from this point on precisely as in the proof of Theorem~\ref{thm:main_result}, the correctness of which follows by Fact~\ref{fac:mod_fac}.

We conclude by pointing out that, in case the randomized construction of Lemma~\ref{lem:isolation} was used, we can perform probability amplification by repeating the final algorithm $2^k$ times to end up in a constant success probability.
\end{proof}

\end{document}